\documentclass{article}[12pt]
\usepackage[a4paper]{geometry}
\usepackage[utf8]{inputenc}
\usepackage[T1]{fontenc}
\usepackage{lmodern}
\usepackage{amsmath,amssymb,amsthm} \allowdisplaybreaks[1]
\usepackage{mathrsfs}
\usepackage{microtype}
\usepackage{latexsym}
\usepackage[colorlinks]{hyperref}
\usepackage{cleveref}
\usepackage{graphicx}
\usepackage{bookmark}
\usepackage{booktabs}
\usepackage{diagbox}
\usepackage{float}
\usepackage{color}
\usepackage{cite}
\usepackage{bm}
\title{Equivalence and Duality of Polycyclic Codes Associated with Trinomials over Finite Fields\thanks{This research is supported by the National Natural Science Foundation of China (12071001) and
the Excellent Youth Foundation of Natural Science Foundation of Anhui Province (1808085J20).}}
\author{Minjia Shi\thanks{smjwcl.good@163.com}, Haodong Lu\thanks{hdlu818@163.com}, Shuang Zhou\thanks{shzhou0815@163.com}, Jiarui Xu\thanks{jrxu39@163.com}, Yuhang Zhu\thanks{yuhangz01@163.com}
\thanks{the Key Laboratory of Intelligent
Computing and Signal Processing, Ministry of Education, School
of Mathematical Sciences, Anhui University, Hefei 230601, China.}}

\newtheorem{thm}{Theorem}[section]
\newtheorem{lm}{Lemma}[section]
\newtheorem{prop}{Proposition}[section]
\newtheorem{cor}{Corollary}[section]
\theoremstyle{definition}
\newtheorem{de}{Definition}[section]
\newtheorem{con}{Conjecture}[section]

\newtheorem{ex}{Example}[section]

\DeclareMathOperator{\ord}{\text{Ord}}

\begin{document}
\maketitle
\begin{abstract}
\noindent
In this paper, several conjectures proposed in \cite{ref1} are studied, involving the equivalence and duality of polycyclic codes associated with trinomials. According to the results, we give methods to construct isodual and self-dual polycyclic codes, and study the self-orthogonal and dual-containing polycyclic codes over $\mathbb{F}_2$. \\
\end{abstract}

\section{Introduction}
Polycyclic codes is a generalization of cyclic codes since the concept of associate polynomials is proposed. Algebraically, a polycylic code is an ideal of $\mathbb{F}_q/\langle f(x)\rangle$ for some $f(x)\in \mathbb{F}_q[x]$, and $f(x)$ is called the associate polynomial of this polycyclic code. Polycyclic codes include cyclic codes, constacyclic codes and etc, as $x^n-1$ can be associated with cyclic codes and $x^n-\lambda$ can be associated with constacyclic codes of length $n$ over $\mathbb{F}_q$. Polycylic codes have been known since $1972$ in \cite{PW}. As is well known, polycyclic
codes are shortened cyclic codes, and conversely shortened cyclic codes are polycyclic (P. 241, \cite{PW}). However, polycyclic codes were never received the same level of attention as cyclic codes and some of their generalizations. In recent years, polycyclic codes has also been studied widely in \cite{ref1,poly1,ref5,poly2,poly3,WSS}, which shows that polycyclic codes begun to receive more and more attention.

In this paper, we mainly solve several conjectures proposed in \cite{ref1}, which mainly around the equivalence and duality of polycyclic codes associated with trinomials. Some kinds of duality of polycyclic codes have been studied, such as the results in \cite{ref3,isodual}, and we will also give some information and results about the duality of polycyclic codes associated with trinomials in this paper.

In Section 2, we mainly give the basic concepts and results of polycyclic codes, they are the basis of this paper and will appear frequently in the following sections. In Section 3, a mapping $\varphi$ will be given to illustrate the relationship between two kinds of polycyclic codes corresponding to $f(x)$ and $f^*(x)$ respectively, where $f(x)$ is trinomial with $f(0)\neq 0$. This particular relationship cannot be obtained simply by choosing polynomials of the same order, which shows that Conjecture 3.5 in \cite{ref1} is wrong. A counter example of Conjecture 3.5 would be given in this section. In Section 4, a conjecture about isodual in \cite{ref1} (i.e. Conjecture 4.5) will be proven, which gives a simple way to construct some isodual polycyclic codes associated with trinomials. However, we will give a counter-example to illustrate conjecture 4.2 in \cite{ref1} is wrong, and point out that this method is not the only way to construct isodual polycyclic codes associated with trinomials over $\mathbb{F}_2$. In Section 5, we study the properties of generator polynomials for self-dual polycyclic codes associated with trinomials, and it turns out that the generators have a unified form $x^n-a$, where $a^2=-1$ in $\mathbb{F}_q$. And these results show that Conjecture 4.7 is correct in \cite{ref1}. In Section 6, we discussed Conjecture 4.10 and Conjecture 4.11 in \cite{ref1}. The results of the discussion show that it is impossible to find self-orthogonal or nontrivial dual-containing polycyclic codes over $\mathbb{F}_2$, and a dual-containing polycyclic code with prime length will be given to show that Conjecture 4.11 is wrong. Section 7 concludes this paper.

\section{Preliminaries}

In this part, we recall the basic definitions and theorems of polycyclic codes.

\begin{de}\cite{ref1}
A linear code $C$ is said to be right polycyclic with respect to $v=(v_0,v_1,\cdots,v_{n-1})$\\ $\in\mathbb{F}_{q}^{n}$, if for any codeword $(c_0,c_1,\cdots,c_{n-1})\in C$, its right polycyclic shift, i.e. $(0,c_0,c_1,\cdots,c_{n-2})+c_{n-1}(v_0,v_1,\cdots,v_{n-1})$ is also a codeword of $C$.
\end{de}

Similarly, we can define $C$ is left polycyclic. If $C$ is both left and right polycyclic, then it is bi-polycyclic. In our paper, when referring to polycyclic codes, we suppose it to be right polycyclic codes.

On the premise that there is no difference between vector $v$ and polynomial $v(x)$, each polycyclic code $C$ of length $n$ is associated with a polynomial $v(x)$ of degree less than $n$, and we say that $C$ is a polycyclic code associated with $x^n-v(x)$. Note that an associate polynomial of a polycyclic code may not be unique, i.e. a polycylic code can have multiple associate polynomials. Furthermore, polycyclic codes associated with $f(x)=x^n-v(x)$ are ideals of the factor ring $\mathbb{F}_q[x]/\langle f(x)\rangle$.

\begin{de}\cite{ref5}
A polycyclic code $C$ over $\mathbb{F}_q$ of length $n$ and dimension $k$ with associate polynomial $v(x)$ has a monic polynomial $g(x)$ of minimal degree $n-k$ that belongs to $C$, which divides $x^n-v(x)$ and is called the generator polynomial of $C$.	
\end{de}

For $\lambda\in \mathbb{F}_q$ and any $v=(v_0,v_1,\cdots,v_{n-1})\in\mathbb{F}_{q}^{n}$, the $\lambda$-constacyclic shift $\tau_\lambda$ on $\mathbb{F}_q^n$ is the shift $\tau_\lambda(v_0,v_1,\cdots,v_{n-1})=(\lambda v_{n-1},v_0,v_1,\cdots,v_{n-2})$. A linear code $C$ is said to be a $\lambda$-constacyclic code if $\tau_\lambda(C)\subseteq C$. Note that the associate polynomials of constacyclic codes have the form $x^n-\lambda$, where the length of constacyclic codes is $n$. Since polycyclic codes are a generalization of cyclic codes, so in the same way as cyclic codes, we can construct a generator matrix for $C$ from its generator polynomial.

\begin{thm} \cite{ref5}
A code $C\in \mathbb{F}_q^n$ is right polycyclic associated with polynomial $f(x)$ if and only if it has a $k\times n$ generator matrix of the form
\[
 G=\begin{pmatrix}
 	g_0 &g_1 &\cdots &g_{n-k} &0 &0 &\cdots &0 \\
 	0 &g_0 &g_1 &\cdots &g_{n-k} &0 &\cdots &0 \\
 	\vdots &\ddots &\ddots &\ddots &\ddots &\ddots &\ddots &\vdots \\
 	0 &\cdots &0 &g_{0} &g_1 &\cdots &g_{n-k} &0 \\
 	0 &\cdots &0 &0 &g_0 &g_1 &\cdots &g_{n-k} \\
 \end{pmatrix}_{k\times n},
 \]
 with $g_{n-k}\neq 0$. In this case $\langle g_0+g_1x+\cdots+g_{n-k}x^{n-k}\rangle$ is an ideal of $\mathbb{F}[x]/\langle f(x) \rangle$.
\end{thm}

In this article, If we say that a polycyclic code $C$ is generated by $g(x)\mid f(x)$, which also means $f(x)$ is the associate polynomial of $C$ and $g(x)$ is the generator polynomial of $C$.

Since we are concerned with the duality of polycyclic codes, it is necessary to mention some definitions of it.

\begin{de}
Two codes $C_1,C_2$ over $\mathbb{F}_q$ with generator matrices $G_1$ and $G_2$ respectively are equivalent if there exists a monomial matrix over $\mathbb{F}_q$ such that $G_1M=G_2$. A monomial matrix is a square matrix with exactly one non-zero entry in each row and each column.
\end{de}

\begin{de}
The dual of $C$, denoted by $C^{\perp}$, is the set of vectors orthogonal to every codeword of $C$ under the Euclidean inner product. A code $C$ is self-dual if $C=C^{\perp}$; $C$ is iso-dual if $C$ is equivalent to $C^{\perp}$; $C$ is self-orthogonal if $C\subseteq C^{\perp}$ and $C$ is dual-containing if $C^{\perp}\subseteq C$.
\end{de}

Finally, here are some notations and theorems that will be used in this paper.

\begin{de} \cite{SAS}
	Let $f\in \mathbb{F}_q[x]$ be a nonzero polynomial. If $f(0)\neq 0$, then the least positive integer $e$ such that $f(x)$ divides $x^e-1$ is called the order of $f$ and denoted by $\ord(f)$. If $f(0)=0$, then $f(x)=x^hg(x)$, where $h\in\mathbb{N}$ and $g\in \mathbb{F}_q[x]$ with $g(0)\neq 0$ are uniquely determined, thus we can define $\ord(f)$ as $\ord(g)$.
\end{de}

\begin{de} \cite{SAS}
Let $g(x)=\sum_{i=0}^{k}g_ix^i$ be a polynomial of degree $k$\:($g_k\neq0$) over $\mathbb{F}_q$. Define the reciprocal polynomial $g^*(x)$ of $g(x)$ by
$g^*(x):=x^k\,g(1/x)=\sum_{i=0}^{k}g_{k-i}x^i$.
\end{de}

\begin{de}
Let $C$ be a code, its reversed code is defined as $C'$ if $(c_{n-1},c_{n-2},\cdots,c_0)\in C'$, where $(c_0,\cdots,c_{n-2},c_{n-1})\in C$.
\end{de}

It's clear that $(g^*(x))^*=g(x)$ and $(C')'=C$.

\begin{thm} \cite{SAS}
Let $C$ be a linear code and let $H$ be a parity-check matrix for $C$. Then the following statements are equivalent:\\
(i) $C$ has distance $d$;\\
(ii) any $d-1$ columns of $H$ are linearly independent and $H$ has $d$ columns that are linearly dependent.
\end{thm}

\section{Equivalence of Polycyclic Codes with Special Associate Polynomials}
In this section, we are ready to analyze Conjecture 3.5 in \cite{ref1}. In fact, Conjecture 3.5 in \cite{ref1} itself is not correct, which can be seen in the counter-example Example 3.1. However, if we change the condition ``$\ord(t_1(x))=\ord(t_2(x))$" to ``$t_2(x)=t_1(0)^{-1}t_1^*(x)$", we can obtain a new conclusion (see Theorem 3.3 below). First, let us recall some claims from \cite{ref1}.

\begin{prop} [Proposition 5.2 in \cite{ref1}]
For any polynomial $f(x)$, $\deg (f(x))\neq \deg(f^*(x))$ if and only if $x\mid f(x)$.
\end{prop}

Since we need $g(x)\mid x^n-ax^i-b$ with $b\neq 0$, using the result of Proposition 3.1, we have $\deg(g(x))=\deg(g^*(x))$.

\begin{thm} [Theorem 4.9 in \cite{ref1}]
Let $C$ be a polycyclic code, then its reversed code $C'$ is generated by the reciprocal polynomial of the generator polynomial of $C$. Furthermore, $C'$ is always equivalent to $C$ because of the isomorphism.
\end{thm}

\begin{thm} [Theorem 6.1 in \cite{ref1}]
$($same order$)$. Let trinomials $f_1(x)=x^n-ax^i-b$ and $f_2(x)=x^n-a'x^{n-i}-b'$, where $a'=-ab^{-1}$ and $b'=b^{-1}$, be mutually reciprocal. Then $\ord(f_1(x))=\ord(f_2(x))$.
\end{thm}

\begin{con} [Conjecture 3.5 in \cite{ref1}]
Let $t_1(x)=x^n-ax^i-b$, $t_2(x)=x^n-a'x^{n-i}-b'$ be such that $\ord(t_1(x))=\ord(t_2(x))$. Let $S_1$ be the set of all polycyclic codes of length $n$ over $\mathbb{F}_q$ associated with $t_1(x)$ and $S_2$ be the set of all polycyclic codes of length $n$ over $\mathbb{F}_q$ associated with $t_2(x)$. Then $S_1$ and $S_2$ are in a one-to-one correspondence where corresponding codes are equivalent to each other.
\end{con}

Now we give a counter example of Conjecture 3.1.

\begin{ex}
	Consider $t_1(x)= x^{10}+x^8+1$ and $t_2(x)=x^{10}+x^2+2$ over $\mathbb{F}_3$. As $\ord(t_1(x))=\ord(t_2(x))=156$, and
	\[\begin{aligned}
		t_1(x) &= (x+1)^2(x+2)^2(x^6+2x^2+1),\\
		t_2(x) &= (x^2+1)^2(x^3+x^2+x+2)(x^3+2x^2+x+1).\\
	\end{aligned}\]
	It is clear that $S_1$ contains 18 linear codes, while $S_2$ contains 12 linear codes. Therefore, no mapping exists to satisfy the requirement since the mapping need to be one-to-one. (The data can be obtained by Magma.)
\end{ex}

\begin{prop}
	$f(x)g(x)=h(x)$ if and only if $f^*(x)g^*(x)=h^*(x)$, where $f(0)\neq 0$ and $g(0)\neq 0$.
\end{prop}

\begin{proof}
	Let $\deg(f(x))=k,\deg(g(x))=n-k,\deg(h(x))=n$, we have
	\[\begin{aligned}
		f(x)g(x)&=h(x),\\
		x^kf(1/x)\cdot x^{n-k}g(1/x)&=x^{n}h(1/x),\\
		f^*(x)g^*(x)&=h^*(x).\\
	\end{aligned}\]
	Since $h(x)=(h^*(x))^*=(f^*(x)g^*(x))^*=(f^*(x))^*(g^*(x))^*=f(x)g(x)$.
\end{proof}

Note that $\ord(t_1(x))=\ord(t_2(x))$ if $t_2(x)=t_1(0)^{-1}t_1^*(x)$ in Theorem 3.2, thus the following theorem (Theorem 3.3) is a special case of Conjecture 3.5. By using Proposition 3.2, if $g(x)$ is a factor of $h(x)$, then $g^*(x)$ is a factor of $h^*(x)$, which means $h(x)$ and $h^*(x)$ have the same number of factors, and combined with Proposition 3.1, $h(x)$ and $h^*(x)$ have the same number of factors of the same degree.

\begin{thm}
	Let $t_1(x)=x^n-ax^i-b$, $t_2(x)=x^n-a'x^{n-i}-b'$ be such that $t_2(x)=t_1(0)^{-1}t_1^*(x)$. Let $S_1$ be the set of all polycyclic codes of length $n$ over $\mathbb{F}_q$ associated with $t_1(x)$ and $S_2$ be the set of all polycycliccodes of length $n$ over $\mathbb{F}_q$ associated with $t_2(x)$. Then $S_1$ and $S_2$ are in a one-to-one correspondence where corresponding codes are equivalent to each other.
\end{thm}

\begin{proof}
	By Proposition 3.2, if $g(x)$ is a factor of $t_1(x)$, then $g^*(x)$ is a factor of $t_2(x)$, and $\deg(g(x))=\deg(g^*(x))$. If the generator of a polycyclic code $C$ is $g(x)$, then $g(0)^{-1}g^*(x)$ is the generator of $C'$. Based on the above facts, define the mapping
	\[\begin{aligned}
	\varphi:\; S_1 &\rightarrow S_2,\\
		C &\mapsto C'.\\
	\end{aligned}\]
$\varphi$ satisfies:\: 1) $\varphi$ is an injection. For $\: C_1,C_2\in S_1$ and $C_1\neq C_2$, let the generators of $C_1$, $C_2$ be $f(x)$, $g(x)$, respectively. If $\varphi(C_1)=\varphi(C_2)$, then $f(0)^{-1}f^*(x)=g(0)^{-1}g^*(x)$, which means $f(0)^{-1}f(x)=g(0)^{-1}g(x)$ and $C_1=C_2$, which is a contradiction. \: 2) $\varphi$ is surjective. For $\:C \in S_2$, let the generator of $C$ be $g(x)$. Since $g(x)\mid t_2(x)$, then $g(0)^{-1}g^*(x)\mid t_1(x)$, which means that there exists $C_0=C'\in S_1$ with the generator $g(0)^{-1}g^*(x)$, i.e. $\varphi^{-1}(C)=C'$.\: 3) $C$ is equivalent to $\varphi(C)=C'$, which is the result of Theorem 3.1.
\end{proof}
\begin{ex}
Let $t_1(x)=x^5+x^4+1$ and $t_2(x)=x^5+x+1$ over $\mathbb{F}_2$, and
\[\begin{aligned}
t_1(x)&=(x^2 + x + 1)(x^3 + x + 1)=f_1(x)f_2(x),\\
t_2(x)&=(x^2 + x + 1)(x^3 + x^2 + 1)=f_1^*(x)f_2^*(x).\\
\end{aligned}\]
Therefore, $S_1=\{C_1,C_2,C_3,C_4\}$, where the generator of $C_1,C_2,C_3,C_4$ are $1,f_1(x),f_2(x),t_1(x)$, respectively. And $S_2=\{C_1',C_2',C_3',C_4'\}$.
\end{ex}
\section{A way to Construct Isodual Polycyclic Codes}

In this section, Conjecture 4.1, 4.2 and 4.5 in \cite{ref1} are studied. Before we discuss these conjectures, here is a useful proposition.

\begin{prop}
Over $GF(q)$, where $q$ is a power of $2$. If a polycyclic code $C$ associated with trinomial $x^{2m}+a^2x^{2i}+b^2$ and generated by $g(x)=x^m+ax^i+b$ is isodual, then the polycyclic code $C_1$ generated by $g_1(x)=x^{km}+ax^{ki}+b\mid x^{2km}+a^2x^{2ki}+b^2$ is also isodual, where $k$ is a positive integer.
\end{prop}
\begin{proof}
In fact, the generator matrix of $C$ is
\[G=\begin{pmatrix}
b &\cdots &a &\cdots &1 &{} &{}\\
{} &\ddots  &{} &\ddots  &{} &\ddots &{} \\
{} &{} &b  &\cdots &a &\cdots &1\\
\end{pmatrix}_{m\times 2m}.
\]
By changing $b$ to $bI_k$, $a$ to $aI_k$, $1$ to $I_k$ and $0$ to $O_k$, we get a new matrix $G_1$ which is the generator matrix of $C_1$. Here, $O_k$ is the all-zero matrix and $I_k$ is the identity matrix. Let $H$ and $H_1$ be the parity-check matrices of $C$ and $C_1$, respectively, then if there exists a monomial matrix $M$ satisfing $GM=H$, then the matrix $M_1$ also satisfies $G_1M_1=H_1$, where $M_1$ is obtained by replacing the elements $r$ in $M$ with $rI_k$, and $M_1$ is also monomial.
\end{proof}

Now we discuss Conjecture 4.5 in \cite{ref1}, which consists of Theorem 4.1, 4.2 and 4.3.

\begin{thm}
Over $GF(q)$, where $q$ is a power of an odd prime. A polycyclic code $C$ associated with trinomial $x^{2m}\pm 2cx^m+c^2$ and generated by $g(x)= x^m\pm c$ is isodual.
\end{thm}
\begin{proof}
Since $g(x)=x^{m}\pm c$, and $g(x)$ is the generator polynomial of $C$, we have a generator matrix $G$ of the form
\[G=(\pm cI_m\mid I_m),
\]
where $I_m$ is an identity matrix. Since $(\pm c^{-1})G=( I_m\mid(\pm c^{-1})I_m)$, then the parity-check matrix of $C$ is $H=((\mp c^{-1})I_m\mid I_m)$. We can define
\[M=\begin{pmatrix}
-c^2I_m &{} \\
{}   &I_m \\
\end{pmatrix},\]
and $M$ is a monomial matrix, which satisfies $HM=G$, and $C$ is isodual.
\end{proof}

\begin{thm}
Over $GF(q)$, where $q$ is a power of 2. A polycyclic code $C$ associated with trinomial $x^{2m}+a^2x^{2i}+b^2$ and generated by $g(x)=x^m+ax^i+b$ is isodual, where $i\mid m$.
\end{thm}
\begin{proof}
Let $m=ki$, according to Propsition 4.1, we consider the polycyclic code generated by $g(x)=x^k+ax+b$. Using row operations, we can change $G$ to the following standard form
\begin{align*}
	G&=\left(\begin{array}{cccc|cccc}
		1&0&\cdots&0\,&\,b^{-1}+(b^{-1}a)^k&b^{-2}a&\cdots&b^{-1}(b^{-1}a)^{k-1}\\
		0&1&\cdots&0\,&\,(b^{-1}a)^{k-1}&b^{-1}&\cdots&b^{-1}(b^{-1}a)^{k-2}\\
		\vdots&\vdots&\ddots&\vdots\,&\,\vdots&\vdots&\ddots&\vdots\\
		0&0&\cdots&1\,&\,b^{-1}a&0&\cdots&b^{-1}
	\end{array}
	\right)=(I_k\mid A).\end{align*}
Then the parity-check matrix of $C$ is
\begin{align*}
	H&=\left(\begin{array}{cccc|cccc}
		b^{-1}+(b^{-1}a)^k&(b^{-1}a)^{k-1}&\cdots&b^{-1}a\,&\,1&0&\cdots&0\\
		b^{-2}a&b^{-1}&\cdots&0\,&\,0&1&\cdots&0\\
		\vdots&\vdots&\ddots&\vdots\,&\,\vdots&\vdots&\ddots&\vdots\\
		b^{-1}(b^{-1}a)^{k-1}&b^{-1}(b^{-1}a)^{k-2}&\cdots&b^{-1}\,&\,0&0&\cdots&1
	\end{array}
	\right)=(A^T\mid I_k).
\end{align*}

We give the following four steps:

Step 1. Reducing the power of $b$ to nonnegative in $A$ and $A^T$ by doing row operations, respectively, we gain $A_1$, $B_1$.

\[A_1=
\begin{pmatrix}
	\begin{array}{cccccc}
		b^{k-1}+a^k & b^{k-2}a&b^{k-3}a^2&\cdots&ba^{k-2}& a^{k-1}\\
		a^{k-1}& b^{k-2} &b^{k-3}a&\cdots & ba^{k-3}&a^{k-2}\\
		a^{k-2}& 0&b^{k-3}&\cdots&ba^{k-4}&a^{k-3}\\
		\vdots&\vdots&\vdots& \ddots &\vdots& \vdots\\
		a^2&0&0&\cdots &b&a\\
		a & 0 &0&\cdots &0& 1
	\end{array}
\end{pmatrix},
\]
\[
B_1=\begin{pmatrix}
	\begin{array}{cccccc}
		b^{k-1}+a^k & ba^{k-1}&b^{2}a^{k-2} &\cdots&b^{k-2}a^2 & b^{k-1}a\\
		a & b & 0&\cdots& 0&0\\
		a^2 &ba&b^2&\cdots &0&0\\
		\vdots & \vdots&\vdots &\ddots& \vdots&\vdots\\
		a^{k-2}&ba^{k-3}&b^2a^{k-4}&\cdots&b^{k-2}&0\\
		a^{k-1} & ba^{k-2}&b^2a^{k-3}&\cdots &b^{k-2}a&b^{k-1}
	\end{array}
\end{pmatrix}.
\]

Step 2. Swaping the $i$-th\:($i\geq 2$) row of $A_1$ with the $(k-i+2)$-th row of $A_1$, we gain $A_2$.
\[
A_2=\begin{pmatrix}
	\begin{array}{cccccc}
		b^{k-1}+a^k&b^{k-2}a&b^{k-3}a^2&\cdots&ba^{k-2}& a^{k-1}\\
		a&0&0&\cdots&0&1\\
		a^2&0&0&\cdots &b&a\\
		\vdots&\vdots&\vdots& \ddots &\vdots& \vdots\\
		a^{k-2}& 0&b^{k-3}&\cdots&ba^{k-4}&a^{k-3}\\
		a^{k-1}& b^{k-2} &b^{k-3}a&\cdots & ba^{k-3}&a^{k-2}
	\end{array}
\end{pmatrix}.
\]

Step 3. Each column of $A_2$ and $B_1$ is multiplied by the inverse of its greatest common factor, respectively, we gain $A_3,B_2$.
\[
A_3=\begin{pmatrix}
	\begin{array}{cccccc}
		b^{k-1}+a^k&a&a^2&\cdots&a^{k-2}& a^{k-1}\\
		a&0&0&\cdots&0&1\\
		a^2&0&0&\cdots &1&a\\
		\vdots&\vdots&\vdots& \ddots &\vdots& \vdots\\
		a^{k-2}& 0&1&\cdots&a^{k-4}&a^{k-3}\\
		a^{k-1}& 1 &a&\cdots & a^{k-3}&a^{k-2}
	\end{array}
\end{pmatrix},
\]
\[
B_2=\begin{pmatrix}
	\begin{array}{cccccc}
		b^{k-1}+a^k & a^{k-1}&a^{k-2} &\cdots&a^2 & a\\
		a & 1 & 0&\cdots& 0&0\\
		a^2 &a&1&\cdots &0&0\\
		\vdots & \vdots&\vdots &\ddots& \vdots&\vdots\\
		a^{k-2}&a^{k-3}&a^{k-4}&\cdots&1&0\\
		a^{k-1} & a^{k-2}&a^{k-3}&\cdots &a&1
	\end{array}
\end{pmatrix}.
\]

Step 4. Swaping the $i$-th\:($i\geq 2$) column of $A_3$ with the $(k-i+2)$-th column of $A_3$, we gain $A_4$.\\
\[
A_4=\begin{pmatrix}
	\begin{array}{cccccc}
		b^{k-1}+a^k& a^{k-1}&a^{k-2}&\cdots&a^2&a\\
		a&1&0&\cdots&0&0\\
		a^2&a&1&\cdots &0&0\\
		\vdots&\vdots&\vdots& \ddots &\vdots& \vdots\\
		a^{k-2}& a^{k-3}&a^{k-4}&\cdots&1&0\\
		a^{k-1}& a^{k-2} &a^{k-3}&\cdots & a&1
	\end{array}
\end{pmatrix}.
\]

After the above steps, $A$ becomes $A_4$ and $A^T$ becomes $B_2$, and we have $A_4=B_2$, which means there exist $S$, $M$, such that $SGM=H$. Thus $C$ is isodual.
\end{proof}

\begin{thm}
Over $GF(q)$, where $q$ is a power of $2$. A polycyclic code $C$ associated with trinomial $x^{2m}+a^2x^{2i}+b^2$ and generated by $g(x)=x^m+ax^i+b$ is isodual, where $i\nmid m$.
\end{thm}
\begin{proof}
Let $m=ki+t\:(1\leq t< i)$. Similar to the proof of Theorem 4.2, we discuss it by dividing it into the following four steps.

Step 1. Reducing the power of $b$ to nonnegative in $A$ and $A^T$ as we do in Case 1, respectively, we gain $A_1,B_1$.

Step 2. Each column of $A_1$ and $B_1$ is multiplied by the inverse of its greatest common factor as we do in Case 1, respectively, we gain $A_2,B_2$.

Step 3. Switching the first $i$ rows of $A_2$ with the center axis of this $i$ rows, and do the same for the last $(k-1)i+t$ rows, we gain $A_3$.

Step 4. Switching the first $i$ columns of $A_3$ with the center axis of them, and do the same for the last $(k-1)i+t$ columns, we gain $A_4$.

Therefore, we also have $A_4=B_2$, which means there exist $S,M$, such that $SGM=H$, and $C$ is isodual.
\end{proof}

Here, we give an examples to illustrate the steps in Theorem 4.3.

\begin{ex}
Considering the polycyclic code $C$ generated by $g(x)=x^5+ax^2+b\mid x^{10}+a^2x^{4}+b^2$ over $\mathbb{F}_{2^m}$, where $2\nmid 5$. We have the generator matrix of the form
\[G=\begin{pmatrix}\begin{array}{ccccc|ccccc}
b &0 &a &0 &0 &1 &0 &0 &0 &0\\
0 &b &0 &a &0 &0 &1 &0 &0 &0\\
0 &0 &b &0 &a &0 &0 &1 &0 &0\\
0 &0 &0 &b &0 &a &0 &0 &1 &0\\
0 &0 &0 &0 &b &0 &a &0 &0 &1\\
\end{array}\end{pmatrix},\]\[
RG=\begin{pmatrix}\begin{array}{ccccc|ccccc}
1 &0 &0 &0 &0 &b^{-1} &b^{-3}a^3 &b^{-2}a &0 &b^{-3}a^2\\
0 &1 &0 &0 &0 &b^{-2}a^2 &b^{-1} &0 &b^{-2}a &0\\
0 &0 &1 &0 &0 &0 &b^{-2}a^2 &b^{-1} &0 &b^{-2}a\\
0 &0 &0 &1 &0 &b^{-1}a &0 &0 &b^{-1} &0 \\
0 &0 &0 &0 &1 &0 &b^{-1}a &0 &0 &b^{-1}\\
\end{array}\end{pmatrix}=(I_5\mid A).
\]
Let $H$ have the form
\[H=\begin{pmatrix}\begin{array}{ccccc|ccccc}
b^{-1} &b^{-2}a^2 &0 &b^{-1}a &0 &1 &0 &0 &0 &0\\
b^{-3}a^3 &b^{-1} &b^{-2}a^2 &0 &b^{-1}a &0 &1 &0 &0 &0\\
b^{-2}a &0 &b^{-1} &0 &0 &0 &0 &1 &0 &0\\
0 &b^{-2}a &0 &b^{-1} &0 &0 &0 &0 &1 &0\\
b^{-3}a^2 &0 &b^{-2}a &0 &b^{-1} &0 &0 &0 &0 &1\\
\end{array}\end{pmatrix}=(A^T\mid I_5).\]
After Step 1 and Step 2, we have
\[A_2=\begin{pmatrix}
b^2 &a^3 &a &0 &a^2\\
a^2 &b &0 &a &0\\
0 &a^2 &1 &0 &a\\
a &0 &0 &1 &0\\
0 &a &0 &0 &1\\
\end{pmatrix},\;\;
B_2=\begin{pmatrix}
b &a^2 &0 &a &0\\
a^3 &b^2 &a^2 &0 &a\\
a &0 &1 &0 &0\\
0 &a &0 &1 &0\\
a^2 &0 &a &0 &1\\
\end{pmatrix}.\]
After Step 3, we can get
\[
A_3=\begin{pmatrix}
a^2 &b &0 &a &0\\
b^2 &a^3 &a &0 &a^2\\
0 &a &0 &0 &1\\
a &0 &0 &1 &0\\
0 &a^2 &1 &0 &a\\
\end{pmatrix},\;\;
B_2=\begin{pmatrix}
b &a^2 &0 &a &0 \\
a^3 &b^2 &a^2 &0 &a\\
a &0 &1 &0 &0 \\
0 &a &0 &1 &0 \\
a^2 &0 &a &0 &1\\
\end{pmatrix}.\]
After Step 4, we can get
\[
A_4=\begin{pmatrix}
b &a^2 &0 &a &0 \\
a^3 &b^2 &a^2 &0 &a\\
a &0 &1 &0 &0 \\
0 &a &0 &1 &0 \\
a^2 &0 &a &0 &1\\
\end{pmatrix},\;\;
B_2=\begin{pmatrix}
b &a^2 &0 &a &0 \\
a^3 &b^2 &a^2 &0 &a\\
a &0 &1 &0 &0 \\
0 &a &0 &1 &0 \\
a^2 &0 &a &0 &1\\
\end{pmatrix}.\]

Combining all the above steps, there exist $R'$ and $M$ such that $R'GM=H$, which means $g(x)$ would generate an isodual polycyclic code.\\
\end{ex}

In \cite{ref1}, the authors show that if $g^2(x)=x^n-ax^i-b$, then over $\mathbb{F}_q$, where $q$ is a power of an odd prime, $g(x)=x^{n/2}\pm c$ with $i=n/2$, $-a=\pm2c$ and $b=-c^2$, and over $\mathbb{F}_q$, where $q$ is a power of 2, $g(x)=x^{n/2}+a'x^{i/2}+b'$ with $(a')^2=-a$ and $(b')^2=-b$. Combining the results of Theorems 4.1, 4.2, and 4.3, we give Theorem 4.4, which is also Conjecture 4.1 in \cite{ref1}

\begin{thm}(Conjecture 4.1 in \cite{ref1})
If a polycyclic code $C$ is generated by $g(x)\mid x^n-ax^i-b$ and $g^2(x)=x^n-ax^i-b$, then $C$ is isodual.
\end{thm}

Theorem 4.4 gives a simple way to construct isodual polycyclic codes, but it's not the only way. Example 4.2 illustrates this, which is also a counter-example of Conjecture 4.2 in \cite{ref1}. Here is the content of Conjecture 4.2. (denotes it as Conjecture 4.1.)

\begin{con} (Conjecture 4.2 in \cite{ref1})
Over $\mathbb{F}_2$, a polycyclic code $C$ generated by $g(x)\mid x^n-ax^i-b$ is isodual if and only if $g^2(x) = x^n-ax^i-b$.
\end{con}

The sufficiency of Conjecture 4.1 can be regarded as a direct corollary (denoted as Corollary 4.1) of Theorem 4.4. As for the necessity of Conjecture 4.1, it is not correct as the following counter example shows.

\begin{cor}
	Over $\mathbb{F}_2$, a polycyclic code $C$ generated by $g(x)\mid x^n-x^i-1$ is isodual if $g^2(x)=x^n-x^i-1$.
\end{cor}

\begin{ex}
Let $C$ be a polycyclic code generated by $g(x)\mid x^{20}+x^{10}+1$ over $\mathbb{F}_2$, where
\[\begin{aligned}
	x^{20}+x^{10}+1&=(x^2+x+1)^2(x^4 + x + 1)^2(x^4 + x^3 + 1)^2,\\
	g(x)&=(x^2 + x + 1)(x^4 + x + 1)^2,\\
	g^2(x)&\neq x^{20}+x^{10}+1.
\end{aligned}
\]
Using Magma, $g(x)$ can also generate an isodual polycyclic code.
\end{ex}

\section{The Structure of Self-dual Polycyclic Codes}
This section is devoted to solving Conjecture 4.7 in \cite{ref1} (see Theorem 5.1 in this section). Let us first give the following lemma and proposition.

\begin{lm} [Lemma 4.6 in \cite{ref1}]
Let $g(x)=a_0+a_1x+\cdots+a_kx^k\mid x^n-ax^i-b$ be the generator polynomial of a polycyclic code $C$ of length $n$, and let $e=Ord(x^n-ax^i-b)$ and $h(x)=b_0+b_1x+\cdots+b_{e-k}x^{e-k}=\dfrac{x^e-1}{g(x)}$. Then a parity-check matrix $H$ of $C$ is given by
\[H=\begin{pmatrix}
	\begin{array}{ccccccc}
	b_{e-k}&b_{e-k-1} &b_{e-k-2}&\cdots &\cdots&\cdots&\cdots\\
	0&b_{e-k} &b_{e-k-1} &b_{e-k-2} &\cdots &\cdots &\cdots \\
	\vdots &\ddots &\ddots &\ddots &\ddots &\ddots  &\vdots \\
	0 &\cdots &0 &b_{e-k} &b_{e-k-1} &b_{e-k-2} &\cdots\\
	\end{array}
\end{pmatrix}_{k\times n},
\]
which is a submatrix of the parity check matrix of the cyclic code of length $e$ generated by $g(x)\mid x^e-1$.
\end{lm}

\begin{prop}
	Let the monic polynomial $g(x) \in \mathbb{F}_q[x]$ with $\deg(g(x))=k>0$. If both $g(x)\mid x^{2k}+1$ and $g(x)\mid x^{2k}-bx^i-c$, where $b,c\neq 0$, then $k=i$, $c+1\neq 0$ and $g(x)=x^k-a$, where $a \neq 0$.
\end{prop}
\begin{proof}
	If $c+1=0$, then $x^{2k}-bx^i+1 = (x^{2k}+1)-bx^i$, which means $\gcd(x^{2k}+1,x^{2k}-bx^i+1)=\gcd(x^{2k}+1,x^i)=1$. Since $g(x)\mid 1$, which means $\deg(g(x))=0$, which is a contradiction. Hence $c+1\neq 0$, we have
	\[\begin{aligned}
		x^{2k}-bx^i-c &\equiv -bx^i+\alpha\;(\bmod\; x^{2k}+1),\\
	\end{aligned}\]
	where $\alpha=-(c+1)\neq 0$. Since $g(x)\mid \gcd(x^{2k}-bx^i-c,x^{2k}+1)$, we have $g(x)\mid \gcd(x^{2k}+1,-bx^i+\alpha)$, i.e. $i\geq k$. If $2k>i>k$, let $2k=i+t$ \:($0< t <k$), then
	\[
	x^{2k}+1 \equiv \alpha b^{-1}x^t+1\; (\bmod \; -bx^i+\alpha).
	\]
	We have $\gcd(x^{2k}+1,-bx^i+\alpha)=\gcd(-bx^i+\alpha ,\alpha b^{-1}x^t+1)$, then $g(x)\mid \gcd(-bx^i+\alpha ,\alpha b^{-1}x^t+1)$, which means $\deg(\alpha b^{-1}x^t+1)=t\geq k$, which is a contradiction. Therefore, $k=i$. If $-bx^k+\alpha \mid x^{2k}+1$, since $\deg(g(x))=k$, then $g(x)=\gcd(x^{2k}+1,-bx^k+\alpha)=x^k-b^{-1}\alpha$. If $-bx^k+\alpha \nmid x^{2k}+1$, then
	\[
	x^{2k}+1 \equiv \gamma \; (\bmod \; -bx^k+\alpha),
	\]
	where $\gamma\neq 0$. Therefore, $\gcd(x^{2k}-bx^i-c,x^{2k}+1)=\gcd(-bx^k+\alpha,\gamma)=1$, which means $g(x)\mid 1$ and $\deg(g(x))=0$. Therefore, $g(x)=x^k-a$, where $a=b^{-1}\alpha=-(c+1)b^{-1}\neq 0$.
\end{proof}

The generators of self-dual polycyclic codes have the same form, which can be easily seen in the following theorem.
\begin{thm} (Conjecture 4.7 in \cite{ref1})
A polycyclic code $C$ generated by the monic polynomial $g(x)\mid x^n-bx^i-c$ is self-dual if and only if $g(x)=x^k-a\mid x^n-bx^k-c$, where $a^2=-1$, $\ord(x^n-ax^i-b)\geq 3k$ and $n=2k$. Also, the minimum distance of all self-dual polycyclic codes is $2$ and they are all actually constacyclic codes.
\end{thm}
\begin{proof}
	We first prove the necessity. Let $\ord(x^n-bx^i-c)=e$ and $g(x)=\sum_{i=0}^kg_ix^i$, where $\deg(g(x))=k$ and $g_k=1$. Since $C$ is self-dual, we need $n-k=k$, i.e. $n=2k$. As $G$ has the form in Theorem 2.1, and let the parity-check matrix of $C$ be $H$, which has the form in Lemma 5.1. Here, let $h(x)=\sum_{i=0}^{e-k}b_{i}x^i$, then $h(x)g(x)=x^e-1$, which means $b_{e-k}=1$ and $b_0 g_0=-1$. Let $G_1=g_0^{-1}G$. In fact, $G_1$ and $H$ has the form
	\[
	G_1=\begin{pmatrix}\begin{array}{ccccc|ccccc}
			1 &g_0^{-1}g_1 &g_0^{-1}g_2 &\cdots &g_0^{-1}g_{k-1} &g_0^{-1} &{}&{}&{}&{}\\
			{} &1 &g_0^{-1}g_1 &\cdots &g_0^{-1}g_{k-2} &g_0^{-1}g_{k-1} &g_0^{-1}&{}&{}&{}\\
			{} &{} &\ddots &\vdots &\vdots &\vdots &\vdots &\ddots &{} &{}\\
			{} &{} &{} &\ddots &\vdots &\vdots &\vdots &\ddots &\ddots &{}\\
			{} &{} &{} &{} &1 &g_0^{-1}g_1 &g_0^{-1}g_2 &\cdots &\cdots &g_0^{-1} \\
	\end{array}\end{pmatrix}_{k\times 2k},
	\]
	\[
H=\begin{pmatrix}\begin{array}{cccc|cccc}
			1 &b_{e-k-1} &b_{e-k-2} &\cdots &\cdots &\cdots &\cdots &\cdots\\
			0 &1 &b_{e-k-1} &b_{e-k-2} &\cdots &\cdots &\cdots &\cdots\\
			\vdots &\ddots &\ddots &\ddots &\ddots &\ddots  &\ddots &\vdots \\
			0 &\cdots &0 &1 &b_{e-k-1} &b_{e-k-2} &\cdots &b_{e-2k}\\
	\end{array}\end{pmatrix}_{k\times 2k}.
	\]
	Assume that $W_1G_1=(I_k\mid G_2)$ and $W_2H=(I_k\mid H_1)$, where $W_1,W_2$ are the product of elementary matrices, respectively. In fact, let $g_j'=g_0^{-1}g_j$, we have
	\[
		W_1=\prod_{r=2}^{k}A_r,\;\;W_2=\prod_{r=2}^{k}B_r.
	\]
Here, change the $r$-th column of $I_k$ to $(-g_{r-1}',-g_{r-2}',\cdots,-g_1',1,0,\cdots,0)^T$, we get $A_r$, and change the $r$-th column of $I_k$ to $(-b_{e-k-r+1},-b_{e-k-r+2},\cdots,-b_{e-k-1},1,0,\cdots,0)^T$, we get $B_r$, where $2\leq r\leq k$.
	Note that $C=C^{\perp}$ if and only if $W_1G_1=W_2H$.
	Checking the $k$-th row of $W_1G_1$ and $W_2H$, we would have
	\[b_{e-k-r}=g_r', {\rm where} \: 1\leq r\leq k.\]
	which means $W_1=W_2$, then $C$ is self-dual if and only if $G_1=H$. Let $h(x)=h_H(x)+h_J(x)$, where $h_H(x)=g_0^{-1}x^{e-2k}g^*(x)$ and $\deg(h_J(x))\leq e-3k$. If $e-3k<0$, since $G_1=H$, we have $h(0)=0$, which contradicts $h(x)g(x)=x^e-1$. Therefore, $e\geq 3k$. As $g(x)h(x)=x^e-1$, we have
	\[
	x^e-1=h(x)g(x)=g_0^{-1}x^{e-2k}g^*(x)g(x)+h_J(x)g(x).
	\]
	Here, $e-2k\geq \deg(h_J(x)g(x))$. If $\deg(h_J(x)g(x))\neq e-2k$, which means there will leave the term $x^{e-2k}$. Hence, $\deg(h_J(x)g(x))=e-2k$, i.e. $h_J(x)g(x)=-x^{e-2k}-1$. Therefore, we have
	\[
	x^{2k}+1=g_0^{-1}g^*(x)g(x).
	\]
	Since $g(x)\mid x^{2k}-bx^i-c$, by Proposition 5.1, we have $i=k$ and $g(x)=x^k-a$. Then
	\[
	-a(x^{2k}+1)=g(x)g^*(x)=(x^k-a)(-ax^k+1)=-a(x^{2k}+1)+(1+a^2)x^{k},
	\]
	i.e. $a^2=-1$. Then $g(x)$ has the form $g(x)=x^k-a$, where $a^2=-1$.
	
	Now we prove the sufficiency. Since $g(x)=x^k-a$ and $a^2=-1$, the generator matrix $G$ of $C$ has the form $G=(-aE_k\mid E_k)$, and $aG=(E_k\mid aE_k)$, then the parity-check matrix of $C$ is $H=(-aE_k\mid E_k)$, which means $G=H$, then $C$ is self-dual. Furthermore, since the first and $(k+1)$-th columns of $H$ are linearly dependent, by Theorem 2.2, the minimum distance of $C$ is $2$. Since $g(x)g^*(x)=g(0)(x^{2k}+1)$, i.e. $g(x)\mid x^{2k}+1$, which means $C$ is actually a constacyclic code since $x^{2k}+1$ is also an associate polynomial of $C$.
\end{proof}

\begin{ex}
Over $GF(5)$, a polycyclic code of length $4$ generated by $g(x)=x^2+2\mid x^4+4x^2+4$ is self-dual.
\end{ex}
\begin{ex}
Over $GF(5)$, a polycyclic code of length $4$ generated by $g(x)=x^2-2\mid x^4-4x^2+4$ is self-dual.
\end{ex}
\section{Self-orthogonal and Dual-containing Polycyclic Codes}
In \cite{ref1}, the authors provided two cases in which self-orthogonal or dual-containing polycyclic codes do not exist. This section is devoted to discussing Conjecture 4.10 and Conjecture 4.11 in \cite{ref1}. We will prove that Conjecture 4.10 in \cite{ref1} is correct (see Theorem 6.1), while we will prove that Conjecture 4.11 is wrong (see Conjecture 6.1) by giving a counter example (i.e. Example 6.1).

\begin{prop}
	Over $\mathbb{F}_2$, $\gcd(x^n+1,x^n+x^i+1) = 1$.
\end{prop}
\begin{proof}
In fact, using the Euclidean algorithm, we have
	\begin{align*}
	 x^n+x^i+1&=(x^n+1)\cdot 1+x^i,\\
	 x^n+1&=x^i\cdot x^{n-i}+1,\\
	 x^i&=1\cdot x^i.
	\end{align*}
Therefore, $\gcd(x^n+1,x^n+x^i+1)=\gcd(x^n+1,x^i)=\gcd(x^i,1)=1$.
\end{proof}

\begin{thm} (Conjecture $4.10$ in \cite{ref1})
There are no self-dual or self-orthogonal polycyclic codes or nontrivial dual-containing polycyclic codes over the binary field, where these codes are associated with trinomials.
\end{thm}
\begin{proof}
(i) For a self-dual polycylic code $C$, it has been proved in \cite{ref1}, so we will not repeat it.

(ii) For a self-orthogonal polycyclic code $C$, let $C$ be generated by $g(x)=\sum_{i=0}^{n-k}g_ix^i\mid x^{n}+x^i+1$ with $g_{n-k}=1$. Let $G$ be the generator matrix of $C$, which has the form in Theorem 2.1.
If we denote the $t$-th row of $G$ as $\boldsymbol{\delta}_t$ for $\: 1\leq t \leq k$, since $C$ is self-orthogonal, then we have
\[\boldsymbol{\delta}_1\cdot \boldsymbol{\delta}_1^T=0 \Leftrightarrow g_0^2+g_1^2+\cdots+g_{n-k}^2=0,\]
which means that $g(x)$ has even terms. Since $g(x)\mid x^{n}+x^i+1$, let $x^{n}+x^i+1=g(x)f(x)$, hence $g(x)f(x)$ also has even terms, even if there are some terms will cancel each other out as cancelling terms will come in pairs over $\mathbb{F}_2$. Then it leads to a contradiction because $x^{n}+x^i+1$ has odd terms. So there are no self-orthogonal polycyclic codes associated with trinomials.

(iii) For a dual-containing polycyclic code $C$, let $C$ be generated by $g(x)=\sum_{r=0}^{s}g_rx^r\mid x^n-x^i-1$, where $g_{s}=1$. We denote the $r$-th row of the generator matrix $G$ by $\boldsymbol{\delta}_r$ as we do in (ii). Here, $g(x)$ has odd terms since $x^n-ax^i-b$ has odd terms. Let $h(x)g(x)=x^e-1$, where $e=\ord(x^n-ax^i-b)$ and $h(x)=\sum_{r=0}^{e-s}h_rx^r$, thus the parity-check matrix $H$ of $C$ has the form in Lemma 5.1.
Let $\boldsymbol{\beta}_r$ be the $r$-th row of $H$, since $C$ is dual-containing, then we have $\boldsymbol{\beta}_r\cdot \boldsymbol{\beta}_r^T=0$ for $\: 1\leq r\leq s$.
Checking $\boldsymbol{\beta}_d$ and $\boldsymbol{\beta}_{d+1}$, we find that
\[\boldsymbol{\beta}_{d+1}\cdot\boldsymbol{\beta}_{d+1}^T+h_{e-s-n+d}^2=\boldsymbol{\beta}_d\cdot\boldsymbol{\beta}_d^T \Rightarrow h_{e-s-n+d}=0,\; {\rm where}\: 1\leq d \leq s-1.\]
Then we can simplify $H$ to the following form
\[H=\begin{pmatrix}
	\begin{array}{ccccccccc}
		h_{e-s}&h_{e-s-1} &\cdots&\cdots &h_{e-n} &0 &\cdots &\cdots &0\\
		0&h_{e-s} &h_{e-s-1} &\cdots&\cdots &h_{e-n} &0 &\cdots&0\\
		\vdots &\ddots &\ddots &\ddots &\ddots &\ddots  &\ddots &\ddots&\vdots \\
		0 &\cdots &0 &h_{e-s} &h_{e-s-1} &\cdots &\cdots &\cdots &h_{e-n} \\
	\end{array}
\end{pmatrix}_{s\times n}.
\]
Here, $n-s\geq s$. Let the first non-zero element of the consequence $h_{e-n},\cdots,h_{e-s}$ are recorded as $h_u$. Since $g_0=g_s=h_{e-s}=1$ and $C^{\perp}\subseteq C$. There are three cases to be discussed.

Case 1. If $u>e-2s$, then $\boldsymbol{\beta}_1$ cannot be represented as a linear combination of $\boldsymbol{\delta}_1,\cdots,\boldsymbol{\delta}_{n-s}$.

Case 2. If $u=e-2s$, then $\boldsymbol{\delta}_1=\boldsymbol{\beta}_1$, then $\boldsymbol{\delta}_1\boldsymbol{\delta}_1^T=0$, which means $g(x)$ has even terms.

Case 3. If $e-n\leq u<e-2s$, let $h(x)=h_H(x)+h_J(x)$, where $h_H(x)=x^u(\sum_{k=u}^{e-s}h_kx^{k-u})$ and $deg(h_J(x))\leq e-s-n$. Here we need $e\geq s+n>3s$, or $h(0)=0$, which leads to a contradiction since $h(x)g(x)=x^e-1$. Since we can write $\boldsymbol{\beta}_1=\sum_{k=1}^{e-2s-u+1}c_k\boldsymbol{\delta}_k$, where $c_1=c_{e-2s-u+1}=1$, we have $(x^{-u}h_H(x))^*=(\sum_{k=1}^{e-2s-u+1} c_kx^{k-1})g(x)=(1+\cdots+x^{e-2s-u})g(x)$. Therefore,
\[\begin{aligned}
h(x)&=x^u\bigg(\sum_{k=1}^{e-2s-u+1} c_kx^{k-1}\bigg)^*g^*(x)+h_J(x),\\
x^e-1=h(x)g(x)&=x^u(1+\cdots+x^{e-2s-u})g^*(x)g(x)+h_J(x)g(x).  \;\;\;\;\;(1)\\
\end{aligned}\]
Since $\deg(x^{u})=u\geq e-n\geq deg(h_J(x)g(x))$, then there will leave the term $x^u$ if $u\neq e-n$ or $e-n\neq \deg(h_J(x)g(x))$, which means the equation $(1)$ is wrong. If $u=e-n=\deg(h_J(x)g(x))$, then $h_J(x)g(x)=x^u-1$. We have
\[\begin{aligned}
x^e-1&=x^u(1+\cdots+x^{e-2s-u})g^*(x)g(x)+x^u-1,\\
x^{n}-1&=(1+\cdots+x^{e-2s-u})g^*(x)g(x).
\end{aligned}\]
Then $g(x)\mid x^n-1$. Since $g(x)\mid x^n-x^i-1$, according to Proposition 6.1, we have $g(x)=1$, which is the generator polynomial of $\mathbb{F}_2^n$, and this is the extreme situation. Thus, there are no nontrivial dual-containing polycyclic codes associated with trinomials.
\end{proof}

\begin{con}(Conjecture 4.11 in \cite{ref1})
A polycyclic code is not dual-containing if the generator polynomial $g(x)\mid x^n-ax^i-b$, where $n$ is prime.
\end{con}

The following example is a counter example of Conjecture 6.1.

\begin{ex}
Over $\mathbb{F}_3$, $C$ is a polycyclic code generated by $g(x)=x+1\mid x^3+x-1$, where $x^3+x-1=(x+1)(x^2-x-1)$. The generator matrix $G$ and the parity-check matrix $H$ are
\[G=\begin{pmatrix}
	1 &1 &0\\
	0 &1 &1\\
\end{pmatrix},\;\;
H=\begin{pmatrix}
	1 &2 & 1\\
\end{pmatrix}.\]
Since $110+011=121$, we have $C^{\perp}\subseteq C$, i.e. $C$ is dual-containing.	\\
\end{ex}
\section{Conclusion}

The main contributions in this paper are as follows.
\begin{enumerate}
\item [(i)] We studied Conjecture 3.5 in \cite{ref1} in Section 3, which is incorrect because of Example 3.1. Example 3.1 illustrates that the polycyclic codes constructed from polynomials of the same order do not always find a particular mapping between them, which needs to be one-to-one and the corresponding codes are equivalent to each other. By modifying the conditions in Conjecture 3.5, we give a mapping $\varphi$ to connect two kinds of polycyclic codes associated with $f(x)$ and $f^*(x)$ respectively, where $f(0) \neq 0$ in Theorem 3.3. Based on equivalence, we only need to study polycyclic codes associate with $f(x)$ to get the structure of polycyclic codes associate with $f^*(x)$.

\item [(ii)] We prove Conjecture 4.1 and Conjecture 4.5 in \cite{ref1}, and give Example 4.2 to show that Conjecture 4.2 in \cite{ref1} is wrong in Section 4. We study isodual polycyclic codes associated with trinomials in this section, and give a useful way to construct isodual polycyclic codes, which are generated by $g(x)$ and associated with $g^2(x)$ in Theorem 4.4 (corresponds to Conjecture 4.1 and 4.5 in \cite{ref1}). However, in Example 4.2, we give an isodual polycyclic code generated by $g(x)\mid x^n+x^i+1$ with $g^2(x)\neq x^n+x^i+1$ over $\mathbb{F}_2$, which means $g^2(x)=x^n+x^i+1$ is just a necessary condition of isodual polycyclic codes over $\mathbb{F}_2$, and over $\mathbb{F}_q$.

\item [(iii)] We prove Conjecture 4.7 in \cite{ref1} in Section 5. We show the unique form of generator of self-dual polycyclic codes associated with trinomials, that is $x^k-a$ and $a^2=-1$ over $\mathbb{F}_q$ in Theorem 5.1 (corresponds to Conjecture 4.7 in \cite{ref1}). And the minimum distance of all self-dual polycyclic codes is $2$. Theorem 5.1 also shows that there is only one way to construct a polycyclic code associated with trinomials.

\item [(iii)] We discussed Conjecture 4.10 and Conjecture 4.11 in Section 6. Theorem 6.1 (corresponds to Conjecture 4.10 in \cite{ref1}) illustrates that it is impossible to find self-dual, self-orthogonal or nontrivial dual-containing polycyclic codes associated with trinomials over $\mathbb{F}_2$. Example 6.1 is a counter example of Conjecture 4.11 in \cite{ref1}, which means that there exist prime length polycyclic codes associated with trinomials over $\mathbb{F}_q$.

\end{enumerate}

To sum up, we have solved all the conjectures in \cite{ref1}.

\end{document}